\documentclass[letterpaper]{article}
\usepackage{uai2019}
\usepackage[margin=1in]{geometry}

\usepackage{times}

\usepackage[utf8]{inputenc}
\usepackage{microtype}
\usepackage{graphicx}
\usepackage{subfigure}
\usepackage{booktabs} 
\usepackage{amsfonts}
\usepackage{tikz}
\usetikzlibrary{fit,positioning}
\usepackage{url}
\usepackage{stmaryrd}
\usepackage{xcolor}
\usepackage{amsmath}
\usepackage{bbm}
\usepackage{amsthm}
\usepackage{dsfont}
\usepackage{bbold}
\usepackage{makecell}
\usepackage{multirow}

\usepackage{natbib}

\newtheorem{prop}{Proposition}
\def\bal#1{\begin{align}#1\end{align}}

\DeclareMathOperator*{\argmax}{argmax}
\newcommand{\E}[1]{\langle#1\rangle}
\newcommand{\Eq}[1]{\langle#1\rangle_q}

\newcommand{\ve}[1]{ {\mathbf{#1}} }

\newcommand{\oli}[1]{\textcolor{red}{#1}}

\newcommand{\W}{\ve{W}}
\renewcommand{\H}{\ve{H}}
\newcommand{\Y}{\ve{Y}}
\newcommand{\C}{\ve{C}}
\renewcommand{\c}{\ve{c}}
\newcommand{\Z}{\ve{Z}}
\newcommand{\N}{\ve{N}}

\title{Recommendation from Raw Data with Adaptive Compound Poisson Factorization}

\author{
Olivier Gouvert,  Thomas Oberlin, Cédric Févotte \\
IRIT, Université de Toulouse, CNRS, France \\
firstname.lastname@irit.fr
}

\begin{document}

\maketitle

\begin{abstract}
Count data are often used in recommender systems: they are widespread (song play counts, product purchases, clicks on web pages) and can reveal user preference without any explicit rating from the user. Such data are known to be sparse, over-dispersed and bursty, which makes their direct use in recommender systems challenging, often leading to pre-processing steps such as binarization. The aim of this paper is to build recommender systems from these raw data, by means of the recently proposed compound Poisson Factorization (cPF). The paper contributions are three-fold: we present a unified framework for discrete data (dcPF), leading to an adaptive and scalable algorithm; we show that our framework achieves a trade-off between Poisson Factorization (PF) applied to raw and binarized data; we study four specific instances that are relevant to recommendation and exhibit new links with combinatorics. Experiments with three different datasets show that dcPF is able to effectively adjust to over-dispersion, leading to better recommendation scores when compared with PF on either raw or binarized data.
\end{abstract}

\section{INTRODUCTION}
Collaborative filtering (CF) techniques have been achieving state-of-the-art performances in recommendation tasks since the Netflix prize \citep{bennett_netflix_2007}. CF is based on feedbacks of users interacting with items. These data can either be explicit (ratings, thumbs up/down) or implicit (number of times a user listened to a song, number of clicks on web pages). In particular, historical data are easy to collect and often in the form of count data. They can be stored in a sparse matrix $\Y$ of size $U\times I$, where each entry of the matrix $y_{ui}$ is the number of times the user $u\in\{1,\dotsc,U\}$ interacts with the item $i\in\{1,\dotsc,I\}$. For the rest of the paper, we will consider the example of users listening to songs without loss of generality.

Matrix factorization (MF) allows to make recommendations using these feedback data \citep{koren_matrix_2009}. The aim of MF is to infer a low-rank approximation of the observations: $\Y \approx \W\H^T$, where $\W$ of size $U\times K$ represents the preferences of users, and $\H$ of size $I\times K$ represents the attributes of items, with $K\ll \min(U,I)$. Therefore, each user or item is represented in the same latent space  by a vector of $K$ latent components. The strength of an interaction between a user and an item is measured by the dot product between their representative latent vectors. 
Among the methods based on MF \citep{lee_learning_1999,lee_algorithms_2001,hu_collaborative_2008,ma_probabilistic_2011,fevotte_algorithms_2011,liang_modeling_2016}, Poisson factorization (PF) \citep{canny_gap:_2004,cemgil_bayesian_2009,gopalan_scalable_2013} has become very popular in CF when using implicit feedbacks. Indeed, PF posits that the data are generated from a Poisson distribution, making it well-adapted for count data. PF has reached state-of-the-art results while having favorable properties. 
(i)~PF down-weighs the effect of the zeros present in the data, by implicitly assuming that the users have a limited budget to distribute among the items \citep{gopalan_scalable_2013}.
(ii)~Algorithms for PF scale with the number of non-zero values in the data, leading to fast inference \citep{cemgil_bayesian_2009}. 
Many variants of PF have been proposed these last years. Hierarchical structures on the latent variables have been explored \citep{ranganath_deep_2015,zhou_poisson_2015,liang_variational_2018}. 
Other works have proposed to use additional information in the model to perform hybrid CF approaches \citep{gopalan_content-based_2014,lu_learning_2018,salah_2018}.

However, in many cases, count data are over-dispersed and bursty \citep{kleinberg_bursty_2003,schein_poisson-gamma_2016}. The Poisson distribution fails to fully describe such data. Its modeling capacities are indeed limited since its mean and variance are equal. To avoid this problem, it is of common use to work with binarized data \citep{gopalan_scalable_2013,liang_modeling_2016}. This pre-processing step is effective in practice but removes the information contained in the non-zero values. 
Recent works have focused on directly using the raw data in order to achieve better representation and recommendation results. In \citep{hu_collaborative_2008,pan_one-class_2008}, the raw data are introduced as weights (confidence), which regularize the MF approximation. Other works try to find generative processes which are able to deal with over-dispersed data. In \citep{zhou_nonparametric_2017}, the author makes use of the negative binomial (NB) distribution, which is a well-known extension of the Poisson distribution \citep{lawless_negative_1987}. He exploits the compound Poisson (cP) representation of the NB distribution to preserve the scalability property of the proposed algorithm. cP structure has further been used in \citep{simsekli_learning_2013,basbug_hierarchical_2016} to model continuous or discrete sparse data, showing an improved description of the non-zero values. 

In this paper, we present novel contributions to discrete compound Poisson factorization (dcPF). dcPF refers to compound Poisson factorization (cPF), as introduced by \citep{basbug_hierarchical_2016}, for discrete data. dcPF posits that the listening counts can be grouped in listening sessions which are somewhat more informative for recommendation. It uses the concept of self-excitation \citep{du_time-sensitive_2015,hosseini_recurrent_2017,khodadadi_continuous-time_2017,zhou_nonparametric_2017}, which describes the idea that a user can listen to a song not merely because of his/her attachment to it, but because of a previous interaction. The contributions of the paper are the following:

$\bullet$ We develop a unified framework for dcPF and study four specific distributions to model self-excitation, called element distributions. We exhibit new links between the choice of this distribution and combinatorics. 

$\bullet$ We provide simple conditions to preserve scalability and to obtain closed-form updates for the inference of the posterior. 

$\bullet$ We show that dcPF is a natural generalization of PF by proving that PF applied to raw data and PF applied to binarized data are two limit cases of dcPF. 

$\bullet$ We discuss the choice of the element distribution and in particular consider a new one in the context of compound Poisson models, the shifted negative binomial distribution. We present new methodology for hyper-parameter estimation and report experiments with three datasets.

The paper is organized as follows. In Section~\ref{sec:preliminaries}, we provide preliminary material about PF and exponential dispersion models (EDM). In Section~\ref{sec:bay_dcpf}, we present Bayesian dcPF and give an intuitive interpretation of the model and its properties. Related works are discussed in Section~\ref{sec:related} and  a scalable variational algorithm is developed in Section~\ref{sec:vbem}. In Section~\ref{sec:expe}, we apply the proposed algorithm to recommendation tasks with three datasets. Section~\ref{sec:conc} concludes the paper and discusses possible perspectives.

\section{PRELIMINARIES} \label{sec:preliminaries}
\paragraph{Poisson factorization.}
PF is based on non-negative matrix factorization (NMF) \citep{lee_learning_1999,lee_algorithms_2001}. Each observation is assumed to be drawn from a Poisson distribution:
\bal{
	y_{ui} \sim \operatorname{Poisson}([\W\H^T]_{ui}), \label{eq:PF}
}
with ${y_{ui}\in\mathbb{N}}=\{0,1,\dotsc,+\infty\}$. The preferences $\W$ and the attributes $\H$ are supposed to be non-negative matrices. Non-negativity induces a constructive \emph{part-based} representation of the observations that is central to so-called \emph{topic models}~\citep{lee_learning_1999,blei_latent_2003}.
Bayesian extensions of PF typically impose that each entry of the matrices $\W$ and $\H$ has a gamma prior. The gamma prior\footnote{ We use the following convention for the gamma distribution: ${\mathcal{G}(x;\alpha,\beta)=x^{\alpha-1}e^{-\beta x}\beta^\alpha \Gamma(\alpha)^{-1}}$ where $\alpha$ is the shape parameter and $\beta$ is the rate parameter.} imposes non-negativity and is known to induce sparsity when the shape parameter is lower than one. This is a desirable property in the sense that it implies that users and items are represented by only a few patterns. Moreover, it is conjugate with the Poisson distribution, which proves convenient for variational inference.

PF has been very popular in the last decade because of its scalability with sparse data. Sparsity is very common in recommender systems, since subsets of users usually interact with only subsets of items from a large catalog. 
Using the superposition property of the Poisson distribution, we can augment the model presented in Equation~(\ref{eq:PF}) as follows \citep{cemgil_bayesian_2009,gopalan_scalable_2013}:
\bal{
	&y_{ui} = \sum_k c_{uik};~c_{uik} \sim \operatorname{Poisson}(w_{uk}h_{ik}).
}
The conditional distribution of this new latent variable follows a multinomial distribution: 
${\c_{ui}|y_{ui} \sim \operatorname{Mult}(y_{ui},\boldsymbol{\phi}_{ui})}$, 
where $\boldsymbol{\phi}_{ui}$ is a vector of size $K$ with entries ${\phi_{uik} = \frac{w_{uk}h_{ik}}{[\W\H^T]_{ui}}}$. The latent variable $\c_{ui}$ is central to state-of-the-art PF algorithms. $y_{ui}=0$ implies that $\c_{ui}={\bf 0}_K$, where ${\bf 0}_K$ is a vector of size $K$ full of zeros. As such, the latent variable $\c_{ui}$ only needs to be estimated for the non-zero values of $\Y$, which ensures scalability provided the data is sparse.

One limitation of the Poisson distribution is that its variance is equal to its mean: $var(y_{ui})=\E{y_{ui}}$. This makes it ill-suited for over-dispersed data. Moreover, when working with raw data, it appears that PF does not correctly weigh the observations and is too sensitive to large values. The Poisson distribution, parametrized by only one parameter, thus appears too restrictive to model both sparse and heavy-tailed data. To circumvent these issues, data binarization is often used as pre-processing  \citep{gopalan_scalable_2013,liang_modeling_2016}, with the loss of information it induces. The goal of dcPF studied in this paper is to preserve the data while accounting for sparsity and over-dispersion in the model. When necessary, we denote by $\Y^b$ the corresponding binary version of the observations, where $y^b_{ui}=\mathbb{1}[y_{ui}>0]$. 

\paragraph{Exponential dispersion model.} \label{sec:EDM}
A central element of cP models is the distribution used to model the self-excitation. In this paper, we will assume that it belongs to the well-studied EDM family \citep{jorgensen_properties_1986,jorgensen_exponential_1987}. It is a convenient choice when dealing with cP models \citep{yilmaz_alpha/beta_2012,simsekli_learning_2013,basbug_hierarchical_2016}, as explained next.
Most discrete random variables can be written in the form of a discrete EDM, denoted by $x\sim ED(\theta,\kappa)$, and defined by 
\bal{
	p(x;\theta,\kappa) = \exp(x\theta- \kappa\psi(\theta))h(x,\kappa),~x\in S_{\kappa},
	} 
where $\theta\in\Theta\subset\mathbb{R}$ is called the natural parameter, $\kappa>0$ is called the dispersion parameter, $\psi(\theta)$ is the log-partition function, $h(x,\kappa)$ is the base measure and $S_{\kappa}$ is the support of the distribution, which depends of $\kappa$. The mean of $x$ is given by $\E{x} = \kappa\psi'(\theta)$ and its variance by $var(x)=\kappa \psi''(\theta)$.
One of the most interesting properties of EDM is the property of additivity. If $x_l \sim ED(\theta,\kappa)$ and $y = \sum_{l=1}^{n}x_l$ with $n\in\mathbb{N}$, then $y\sim ED(\theta,n\kappa)$. By convention, we assume that $ED(\theta,0)$ is a Dirac distribution in~$0$.

\section{BAYESIAN DISCRETE COMPOUND POISSON FACTORIZATION} \label{sec:bay_dcpf}

\subsection{MODEL DESCRIPTION} \label{sec:model}

We consider the framework proposed by \citep{basbug_hierarchical_2016}. The generative model of the observations $\Y$ is given by
\bal{
	&w_{uk} \sim \mathcal{G}(\alpha^W,\beta^W_u),~h_{ik} \sim \mathcal{G}(\alpha^H,\beta^H_i), \\
	&n_{ui} \sim \operatorname{Poisson}\left([\W\H^T]_{ui}\right), \\
	&x_{l,ui} \sim ED(\theta,\kappa),~\forall l\in \{1,\dotsc,n_{ui}\}, \\
	&y_{ui} = \sum_{l=1}^{n_{ui}}x_{l,ui}.
}
We here specifically assume that $x_{l,ui}$ is a discrete random variable with support equal to $\mathbb{N}^*=\mathbb{N}\setminus\{0\}$. 
The rate parameters of the gamma priors are treated as deterministic parameters estimated by maximum likelihood (ML). $\boldsymbol{\beta}^W\in\mathbb{R}^U_+$ expresses the activity level of the users and $\boldsymbol{\beta}^H\in\mathbb{R}^I_+$ expresses the popularity of the items. A Bayesian treatment of these parameters is also possible and is considered in \citep{gopalan_scalable_2013}. 
Using the additivity property of EDM, we can easily marginalize the latent variables $x_{l,ui}$, leading to:
	${y_{ui} \sim ED(\theta,n_{ui}\kappa)}$.
Compared to PF, this additional stage in the generative process allows for a flexible description of the observations. There are two additional parameters $\{\theta,\kappa\}$ which control the variance and tail of the distribution.

\paragraph{Interpretation.}
In this paragraph, we will suppose that a user/item pair is fixed. For conciseness, we will omit the corresponding indices $_{ui}$. dcPF introduces new latent variables $n$ and $\{x_{l}\}_l$ for $l\in\{1,\dotsc,n\}$. The latent variable $n$ represents the number of listening sessions the user has had for the song. During each session, indexed by $l$, the user listened to the song a number of times $x_{l}$ which is greater or equal to one. The latent variable $x_l$ models the self-excitation induced by a listening interaction. This concept has been used in \citep{zhou_nonparametric_2017,hosseini_recurrent_2017}. Thus, a user can listen to a song, not merely because he/she likes it, but because of a previous listening/excitation. For example, a user can have a summer crush for a song and may listen to it on repeat. The first listening reflects the interest of the user for this song, whereas the following listenings are the consequence of the first one and reflect a short-term behavior. Therefore these listening counts can be grouped in a few listening sessions that will be more able to represent long-term preferences. 
Finally, the observed variable $y$ is just the aggregation of all the listening counts from all the sessions. The variable $n$ can be viewed as a way to partition the observation $y$ in a smaller number of sessions. This number $n$ better reflects the preferences of the user, since it is deprived of the notion of self-excitation which artificially inflates the number of listening counts. 
In the following, we will denote by $\N\in\mathbb{N}^{U\times I}$ the \emph{exposition matrix} with entries $[\N]_{ui}=n_{ui}$.

\paragraph{Joint log-likelihood.} \label{sec:loglik}
The joint log-likelihood of the observations $\Y$ and of the latent variables $\N$, $\W$ and $\H$ can be written as follows:
\bal{
	\log p(\Y,&\N,\W,\H) = \underbrace{\log p(\Y|\N;\theta,\kappa)}_{\text{Mapping}} \\
	&+ \underbrace{\log p(\N|[\W\H^T])}_{\text{PF structure}}  + \underbrace{\log p(\W,\H)}_{\text{Regularization}}. \notag \label{eq:loglik}
}
We can decompose this log-likelihood into three terms: a probabilistic mapping term corresponding to the compound structure of the observations, a term corresponding to the PF structure on the latent variable $\N$ and a regularization term induced by the gamma priors. Contrary to PF, the factorization is placed on the latent variable $\N$ instead of the data itself, allowing for more flexibility. Going back to our interpretation, this latent variable is more likely to inform on user preference than $\Y$. The mapping term can be viewed as a distortion of the true observations, making them ``more factorizable'' than the raw observations. Therefore, this additional term allows to avoid strong pre-processing stages (such as binarization), letting the data choose their ``own distortion''.

\paragraph{Scalability and tractability.} \label{sec:prop}
By imposing that users will listen to a song at least one time during each session, i.e., $x_{l,{ui}}\in \mathbb{N}^*$, two important properties can be deduced.

First, we have the following equivalence: $y_{ui}=0 \Leftrightarrow n_{ui}=0$. In other words, the observed listening count is equal to zero if and only if the number of listening sessions is equal to zero. Therefore, the latent variable $\N$ is partially known and has the same zeros as $\Y$. Thanks to this, we preserve the scalability property of PF (cf Section~\ref{sec:preliminaries}). Moreover, we have that:
\bal{
	\mathbb{P}(y_{ui}=0)=\mathbb{P}(n_{ui}=0)=e^{-[\W\H^T]_{ui}}.
}
The latent variables $\W$ and $\H$ control the sparsity of the matrix $\Y$, while the element distribution and its parameters $\{\theta,\kappa\}$ only focus on the representation of non-zero values.

The second interesting property is that $n_{ui}\leq y_{ui}$. Therefore, given an observation $y_{ui}$, we know that $n_{ui}$ can only take a finite number of values, bounded by $y_{ui}$ (in particular, $y_{ui}=1 \Rightarrow n_{ui}=1$). This provides efficient means of calculation for the latent variable $\N$ during inference. 

\subsection{EXAMPLES OF ELEMENT DISTRIBUTIONS} \label{sec:examples}

\begin{table*}[ht!] \small
	\caption{Examples of four discrete element distributions. Notation: $\mathbb{R_-^*}=(-\infty;0)$.}
	\begin{center}
	\begin{tabular}{llllllllll}
		\toprule
		Distribution & $\theta$ & $\Theta$ & $\theta^\text{raw}$ & $\theta^\text{bin}$ & $\kappa$ & $\psi(\theta)$ & $h(x,\kappa) $ \\
		\hline
		$x_l \sim \operatorname{Log}(p)$  & $\log(p)$ & $\mathbb{R_-^*}$ & $-\infty$ & 0 & $1$ & $\log(-\log(1-e^\theta))$ & $\frac{x!}{\kappa!}St_1(x,\kappa)$\\
		$x_l \sim \operatorname{ZTP}(p)$ & $\log(p)$ & $\mathbb{R}$ & $-\infty$ & $+\infty$ & $1$ & $\log(e^{e^\theta}-1)$ &  $\frac{x!}{\kappa!}St_2(x,\kappa)$ \\
		$x_l \sim \operatorname{Geo}(1-p)$ & $\log(p)$ & $\mathbb{R_-^*}$ & $-\infty$ & $0$ & $1$ & $\log(\frac{e^{\theta}}{1-e^{\theta}})$ & $\frac{x!}{\kappa!}St_3(x,\kappa)$ \\
		$x_l-1 \sim \operatorname{NB}(a,p)$ & $\log(p)$ & $\mathbb{R_-^*}$ & $-\infty$ & $0$ & $(1,a)^T $ & $(\theta,-\log(1-e^{\theta}))^T$ & $ \frac{\Gamma(x-\kappa_1+\kappa_2)}{\Gamma(x-\kappa_1+1)\Gamma(\kappa_2)}$ \\
		\bottomrule
	\end{tabular}
	\end{center}
	\label{tab:examples}
\end{table*} 

\begin{figure}[t!]
	\centerline{\includegraphics[height=3.3cm]{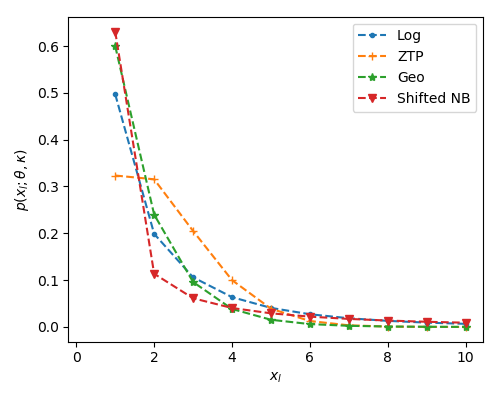} 
	\includegraphics[height=3.3cm]{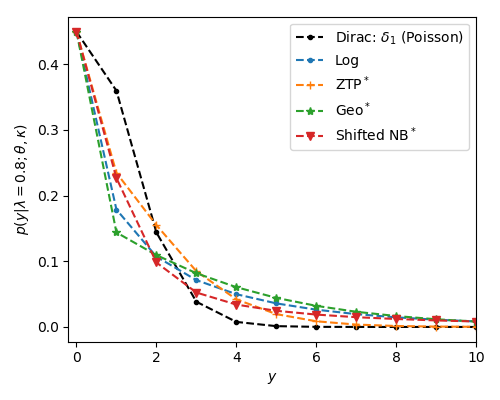}}
	\caption{On the left, the p.m.f. of the four element distributions presented in Section~\ref{sec:examples}. On the right, the p.m.f. of the marginalized distribution of the observations, using the same four element distributions. The values of the natural and dispersion parameters are those presented in Table~\ref{tab:score_tps}. The p.m.f. marked with a $^*$ are not available in closed form and are represented through a histogram of simulated values.}
	\label{fig:pmf}
\end{figure}

\paragraph{Distributions based on Stirling numbers.} \label{sec:examples_stirling}
In this paragraph we focus on three particular distributions: the logarithmic distribution \citep{quenouille_relation_1949}, denoted by $x_l\sim Log(p)$; the zero-truncated Poisson (ZTP) distribution, denoted by $x_l\sim ZTP(p)$; the (shifted) geometric distribution, denoted by $x_l\sim Geo(1-p)$.\footnote{ We use the following convention for the (shifted) geometric distribution: $Geo(x;p)=(1-p)^{x-1}p$, with $x\in\mathbb{N}^*$.} Examples of probabilistic mass functions (p.m.f.) of the four considered element distributions are displayed on Figure~\ref{fig:pmf}.

These three distributions can be written in the form of a discrete EDM with dispersion parameter $\kappa = 1$ and  support  $\mathbb{N}^*$. Their base measure is given by $h(x_l,\kappa) = \frac{x_l!}{\kappa!}St_j(x_l,\kappa)$, where $St_j(x_l,\kappa)$ is the unsigned Stirling number of one of the three kinds ($j\in\{1,2,3\}$), see Table~\ref{tab:examples}. See \citep{johnson_univariate_2005} for more details. It is of particular interest when analyzing the distribution of $y|n$. This conditional distribution is also a discrete EDM with: $h(y,\kappa n) = \frac{y!}{n!} St_j(y,n).$
 
The Stirling numbers of the three kinds are three different ways to partition $y$ elements into $n$ groups \citep{riordan_introduction_2012} (graphical illustrations are given in the supplementary material):

$\bullet$ The Stirling number of the first kind corresponds to the number of ways of partitioning $y$ elements into $n$ disjoints cycles. It can be obtained thanks to a recurrence formula: 
		$St_1(y+1,n) = y~St_1(y,n) + St_1(y,n-1)$.

$\bullet$ The Stirling number of the second kind corresponds to the number of ways of partitioning $y$ elements into $n$ non-empty subsets. It can be calculated in closed form:
		$St_2(y,n) = \frac{1}{n!}\sum_{j=0}^{n} (-1)^{n-j} \binom{n}{j} j^y$.
When $y$ is too large, its exact computation can suffer from numerical issues, though reasonable and stable approximations are available \citep{bleick_asymptotics_1974}.

$\bullet$ The Stirling number of the third kind (also known as Lah number) corresponds to the number of ways of partitioning $y$ elements into $n$ non-empty ordered subsets. It is given by:
		$St_3(y,n) = \binom{y-1}{n-1} \frac{y!}{n!}$.
Its definition is particularly well adapted if we assume that the grouping results from temporal phenomena.

\paragraph{Shifted negative binomial.}
The dispersion parameter for the three distributions presented in the previous examples is fixed and equal to one. We now introduce a new distribution, referred to as \emph{shifted NB} distribution which is parametrized by two parameters: $x_l-1 \sim NB(a,p)$, whose shape parameter $a$ controls the long tail of the distribution, and $p\in (0,1)$ is the probability parameter. The shifted NB is a shifted EDM, which does not exactly fall into the EDM family. However, the conditional distribution $y|n$ can still be written as: $p(y|n;\theta,\kappa) = \exp(y\theta- n\kappa^T\psi(\theta))h(y,n\kappa)$, where $y\in \{n,\dotsc,+\infty\}$, $\kappa = (\kappa_1,\kappa_2)^T= (1,a)^T$ and $\psi(\theta)=(\theta,-\log(1-e^{\theta}))^T$. Note that $\kappa$ and $\psi(\theta)$ are now vectors of dimension 2. Parameter $\kappa_1$ controls the shifting operation, and is fixed to one to ensure that the support of $x_l$ is $\mathbb{N}^*$. Shifted NB encompasses two particular cases: the classical NB distribution ($\kappa_1=0$), and the geometric distribution ($\kappa_1=\kappa_2=1$). 

For each distribution, the resulting marginalized distribution of $y$ is displayed on Figure~\ref{fig:pmf} and is compared to the Poisson distribution (which is a cP distribution with a Dirac as element distribution). We can see on this figure that all marginalized distributions have the same mass in $0$ (see Section~\ref{sec:prop}) but are different otherwise.

\subsection{A TRADE-OFF BETWEEN RAW AND BINARIZED DATA } \label{sec:tradeoff}

In this section, we show that dcPF generalizes PF in the sense that \oli{it} includes PF applied to raw and binarized data as limit cases. For a given dispersion parameter $\kappa$, the natural parameter $\theta$ controls the level of information contained in the observations $\Y$:

$\bullet$ When $\theta$ tends to a limit $\theta^\text{raw}$, dcPF becomes equivalent to PF (with original raw data).

$\bullet$ When $\theta$ tends to a limit $\theta^\text{bin}$, the posterior inference of dcPF becomes equivalent to the posterior inference of PF applied to the binarized data. In other words, performing dcPF (with original raw data) becomes equivalent to performing PF on binarized data. Note that the marginal distribution of the observations and the distribution of $y|n$ are both improper distributions, but the posterior distribution is still well-defined~\citep{robert_bayesian_2007}.

$\bullet$ Between $\theta^\text{raw}$ and $\theta^\text{bin}$, $\theta$ controls the degree of \emph{implicit distortion} of the observations.

Our results are formalized in the two following propositions. The proofs are left to the supplementary material.
\begin{prop}
	If there exists $\theta^\text{raw}$ such that $ \lim_{\theta\to\theta^\text{raw}}\kappa^T \psi(\theta) = -\infty$, then the posterior of dcPF tends to the posterior of PF as $\theta$ goes to $\theta^\text{raw}$.
\end{prop}
\begin{prop}
	If there exists $ \theta^\text{bin}$ such that $\lim_{\theta\to\theta^\text{bin}}\kappa^T \psi(\theta) = +\infty$, then the posterior of dcPF tends to the posterior of PF applied to binarized data as $\theta$ goes to $\theta^\text{bin}$, i.e.: 
		${\lim_{\theta\to\theta^\text{bin}} p(\W,\H|\Y) = p(\W,\H|\N=\Y^b)}$.
\end{prop}

The four distributions described in Section~\ref{sec:examples} respect the assumptions of both propositions. The limit cases of the natural parameter $\theta$ are given in Table~\ref{tab:examples}. 

It is of particular interest to learn the natural parameter $\theta$ since its choice characterizes the data. If $\theta$ is close to $\theta^\text{raw}$, the observations do not need to be distorted and PF on raw data is effective. If $\theta$ is close to $\theta^\text{bin}$, the non-zero observations of $\Y$ are non-informative and binarization is welcome. In between these extremes, dcPF takes full power and acts as an implicit distortion. Thus, the value of $\theta$ gives an indication on the gain brought by dcPF as compared to PF.

\section{RELATED WORKS} \label{sec:related}
\paragraph{Negative binomial factorization.}
An extension of the Poisson distribution known to model over-dispersion is the NB distribution. The NB distribution depends on two parameters: a shape parameter and a probability parameter $p$. In \citep{zhou_nonparametric_2017}, the author introduces NB matrix factorization, in which he posits that the shape parameter is low-rank, i.e.: $y_{ui} \sim NB([\W\H^T]_{ui},p)$.
To preserve scalability of the proposed Gibbs algorithm, the author uses the cP representation of the NB \citep{quenouille_relation_1949,fisher_relation_1943}: 
	$n_{ui} \sim \operatorname{Poisson}(-[\W\H^T]_{ui}\log(1-p))$ 
	and $y_{ui} \sim \operatorname{SumLog}(n_{ui},p)$
where $\operatorname{SumLog}(n,p)$ is the sum of $n$ identical and independent logarithmic distributions. In this case, the conditional distribution of the $n_{ui}$ is also known: $n_{ui}|y_{ui} \sim CRT(y_{ui},[\W\H^T]_{ui})$, where CRT is the number of opened tables in a Chinese restaurant process (CRP). 
An important difference with our framework is that the parameter $p$ there controls both the sparsity of $\Y$ and the distribution of the non-zero values. This introduces a coupling between the factorization $\W\H^T$ and the parameter $p$ which leads to a more difficult interpretation in the context of recommendation.

\paragraph{Compound Poisson models.}
In \citep{basbug_hierarchical_2016}, the authors introduce cPF which is well-adapted for continuous or discrete sparse data. For discrete data, the authors present four different distributions but only one (the ZTP distribution) with support $\mathbb{N}^*$. Note that, if $\mathbb{P}(x_{l,ui}=0)>0$ then the latent variable $\N$ is completely unknown and the scalability property does not hold anymore (unless the hypothesis $y_{ui}=0\Rightarrow n_{ui}=0$ is arbitrarily imposed during the inference). In terms of inference, \citep{basbug_hierarchical_2016} describe a stochastic variational inference algorithm that is shown to perform well in terms of log-likelihood computed from held-out data. We will instead evaluate the performance of dcPF with recommendation metrics.

In \citep{yilmaz_alpha/beta_2012,simsekli_learning_2013}, a cP structure with a gamma element distribution is used to represent the Tweedie distribution. The Tweedie distribution is the distribution induced by the $\beta$-divergence with $\beta\in {(01)}$ \citep{fevotte_algorithms_2011}. One importance difference with our setting, besides the fact that the Tweedie distribution is continuous, is that the authors impose that the model is mean-parametrized, i.e., $\E{y_{ui}} = [\W\H^T]_{ui}$. This is not the case with cPF since by construction: $\E{y_{ui}}\geq \E{n_{ui}} = [\W\H^T]_{ui}$.

\paragraph{Weighted MF.} 
In \citep{hu_collaborative_2008}, the authors develop a framework for implicit feedbacks. Implicit data are inherently noisy and may not reflect a direct preference of a user for an item, but rather a confidence in the interaction. In this context, implicit feedbacks can be transformed and incorporated as weights in the cost function which is defined as:
\bal{
	C(\W,\H) = \sum_{ui} \omega_{ui} \|y^b_{ui}-[\W\H^T]_{ui} \|_2^2 + \mu R(\W,\H), \label{eq:weight} \notag
}
where $\omega_{ui}=f(y_{ui})$ is the confidence that can be brought in the binary observation $y^b_{ui}$, $f$ is a fixed mapping function, $R(\W,\H)$ is a regularization term and $\mu$ is an hyper-parameter. Here, the mapping function $f$ between the raw data and the confidence is deterministic. Note that some other works focused on introducing probabilistic weights in the data fitting term \citep{liang_modeling_2016,wang_general_2015} but these weights are learned regardless of the raw data.
As discussed previously, dcPF encompasses the raw observations via an additional probabilistic mapping term. This term can also be viewed as a probabilistic confidence term, combining the two latter approaches. Indeed, large listening counts $y_{ui}$ will often lead to a large number of sessions $n_{ui}$, exhibiting a strong confidence in this observation. Nevertheless, this mapping is not deterministic and as such more flexible and robust.

\section{VARIATIONAL BAYES EXPECTATION-MAXIMIZATION} \label{sec:vbem}

In this section, we develop a variational Bayes expectation-maximization (VBEM) algorithm. We denote by ${\Z=\{\N,\C,\W,\H\}}$ the set of latent variables and by $\Phi=\Phi_1\cup\Phi_2$ the set of parameters, with $\Phi_1=\{\theta,\kappa\}$ and ${\Phi_2=\{\alpha^W,\boldsymbol{\beta}^W,\alpha^H,\boldsymbol{\beta}^H\}}$. The aim of this algorithm is to estimate both the posterior $p(\Z|\Y;\Phi)$ and the parameters~$\Phi$. 

\subsection{VARIATIONAL INFERENCE}

Bayesian inference revolves around the characterization of the posterior distribution $p(\Z|\Y;\Phi)$. Unfortunately, this posterior is intractable in our case. Variational inference (VI) \citep{jordan_introduction_1999,blei_variational_2017} consists in approximating this intractable posterior by a simpler distribution $q$ parametrized by its own parameters $\tilde\Phi$, called variational parameters. Thus, the aim of VI is to minimize the Kullback-Leibler divergence between the true and approximate distributions with respect to (w.r.t.) the variational parameters. In practice, it is simpler to maximize the so-called expected lower bound (ELBO), which is an equivalent problem.
A common choice is to assume $q$ to be factorizable (mean-field approximation):
\bal{
	q(\Z) = \prod_{ui} q(n_{ui},\c_{ui}) \prod_{uk}q(w_{uk})\prod_{ik}q(h_{ik}).
}
Though not explicitly shown for conciseness, the variational distribution of each parameter is governed by its own set of parameters (over which optimization takes place). Note that we choose the latent variables $n_{ui}$ and $\c_{ui}$ to remain coupled. We can further decompose the variational distribution of these variables as: $q(n_{ui},\c_{ui}) = q(\c_{ui}|n_{ui})q(n_{ui})$. 

The ELBO can be calculated as follow:
\bal{
	\mathcal{L}(q,\Phi) =&  \Eq{\log p(\Y|\N;\Phi_1)} 
	+ \Eq{\log p(\N,\C|\W,\H)} \notag\\
	& + \Eq{\log p(\W,\H;\Phi_2)} + \mathcal{H}(q),
}
where $\Eq{x}$ is the expectation of the variable $x$ w.r.t. the variational distribution $q$ and $\mathcal{H}(q)$ is the entropy of the distribution $q$.

\paragraph{Coordinate ascent VI.} \label{CAVI}
We use a coordinate ascent for VI (CAVI) algorithm to maximize the ELBO. The CAVI algorithm consists of sequentially optimizing each of the variational parameters while keeping the others fixed. It can be shown that mean-field variational inference naturally leads to the following choice of variational distributions \citep{bishop_pattern_2006}, parametrized by~${\tilde \Phi = \{\boldsymbol{\Lambda},\tilde{\boldsymbol{\alpha}}^W,\tilde{\boldsymbol{\beta}}^W,\tilde{\boldsymbol{\alpha}}^H,\tilde{\boldsymbol{\beta}}^H\}}$:
\bal{
	&q(w_{uk}) = \mathcal{G}(\tilde\alpha^W_{uk},\tilde\beta^W_{uk}),~q(h_{ik}) = \mathcal{G}(\tilde\alpha^H_{ik},\tilde\beta^H_{ik}), \\
	&q(\c_{ui}|n_{ui}) = \operatorname{Mult}\left(n_{ui},\left\{\frac{\Lambda_{uik}}{\Lambda_{ui}}\right\}_k\right), \notag\\
	&q(n_{ui}=n) = \frac{1}{Z_{ui}} (r_{ui})^n \frac{h(y_{ui},n\kappa)}{n!},\forall n\in\{1,\dotsc,y_{ui}\}, \notag
}
where $\Lambda_{ui}=\sum_k \Lambda_{uik}$, $r_{ui} = \Lambda_{ui}e^{-\kappa\psi(\theta)}$ and $Z_{ui} = \sum_{n=1}^{y_{ui}} (r_{ui})^n \frac{h(y_{ui},n\kappa)}{n!}$ is a normalization constant. 

\paragraph{Update rules.} CAVI leads to the following set of iterative update rules:
\bal{
	&\Lambda_{uik} \leftarrow \exp\left(\Eq{\log w_{uk}} + \Eq{\log h_{ik}}\right); \\
	&\tilde\alpha^W_{uk} \leftarrow \alpha^W + \sum_i \Eq{n_{ui}}\dfrac{\Lambda_{uik}}{\Lambda_{ui}}; \tilde\beta^W_{uk} \leftarrow \beta^W_u + \sum_i \Eq{h_{ik}} \notag\\
	&\tilde\alpha^H_{ik} \leftarrow \alpha^H + \sum_u \Eq{n_{ui}}\dfrac{\Lambda_{uik}}{\Lambda_{ui}}; \tilde\beta^H_{ik} \leftarrow \beta^H_i + \sum_u \Eq{w_{uk}}. \notag
}
When $x\sim \mathcal{G}(\alpha,\beta)$, $\E{x} = \frac{\alpha}{\beta}$ and $\E{\log x} = \Psi(\alpha) - \log\beta$,  where $\Psi$ is the digamma function. The statistic $\Eq{n_{ui}}$ is available in closed form\footnote{When choosing the logarithmic distribution as the element distribution we have: ${\Eq{n_{ui}} = r_{ui}\big(\Psi(y_{ui}+r_{ui}) - \Psi(r_{ui})\big)}$ \citep{zhou_nonparametric_2017}.} 
thanks to the properties of Section~\ref{sec:prop}. If $y_{ui}=0$ then $\Eq{n_{ui}} = 0$, otherwise $\Eq{n_{ui}} = \sum_{n=1}^{y_{ui}} n~q(n_{ui}=n)$.

As expected, we recover, as limit cases, the algorithms for PF \citep{gopalan_scalable_2013} applied to raw data if ${\Eq{n_{ui}}=y_{ui}}$ and to binarized data if ${\Eq{n_{ui}}={\mathbb{1}[y_{ui}>0]}}$. The algorithm is stopped when the relative increment of the ELBO gets lower than a value $\tau$.

\subsection{PARAMETERS ESTIMATION}
\paragraph{Activity and popularity parameters.}
Optimizing the ELBO w.r.t. the parameters $\Phi_2$ is equivalent to solving the sub-problem: $\argmax_{\Phi_2} \Eq{\log p(\W,\H;\Phi_2)}$. In this article, we suppose that the shape parameters $\{\alpha^W,\alpha^H\}$ are known and we want to optimize only the rate parameters $\{\boldsymbol{\beta}^W,\boldsymbol{\beta}^H\}$. The interested reader is referred to \citep{cemgil_bayesian_2009} and \citep{zhou_negative_2015} for the details of ML and Bayesian estimation of the shape parameter of a gamma distribution. ML leads to the following updates for both activity and popularity parameters: 
\bal{
	\beta^W_u \leftarrow \dfrac{\sum_k \Eq{w_{uk}}}{K\alpha^W };~
	\beta^H_i \leftarrow \dfrac{\sum_k \Eq{h_{ik}}}{K\alpha^H }.
}
\paragraph{Natural parameter.} \label{sec:nat_param_est}
Optimizing the ELBO w.r.t. the natural parameter $\theta$ is equivalent to maximizing: 
	$\Eq{\log p(\Y|\N;\Phi_1)} = \sum_{ui} \left( y_{ui}\theta- \Eq{n_{ui}}\kappa\psi(\theta) \right) + cst, $
where $cst$ is a constant w.r.t. the natural parameter $\theta$. It leads to the following equation:
\bal{
	\sum_{ui} y_{ui} - \sum_{ui} \Eq{n_{ui}} \kappa\psi'(\theta) = 0.
}
In the case of the shifted NB distribution, $\psi(\theta)\in\mathbb{R}^2$ and $\psi'$ corresponds to its gradient. The solution of this equation in known in closed form for geometric and shifted NB distributions. 
We implement a Newton-Raphson algorithm to solve it for logarithmic and ZTP distributions.

\paragraph{Dispersion parameter for shifted NB.}
When choosing the shifted NB as the element distribution, we have:
\bal{y_{ui}-n_{ui}\sim NB(\kappa_2 n_{ui},e^\theta). }
Optimizing the ELBO w.r.t. the  parameter $\kappa_2$  which controls the long-tail of the NB distribution is not straightforward. The main issue is that it involves a term of the form $\Eq{h(y_{ui},n_{ui}\kappa)}$ that is computationally expensive to optimize.
Therefore, we augment the model like in \citep{zhou_nonparametric_2017}, with a latent variable: 
$m_{ui}|y_{ui},n_{ui}\sim CRT(y_{ui}-n_{ui},n_{ui}\kappa_2)$ (if $y_{ui}=0$ then $m_{ui}=0$).
In this augmented model, the optimization of $\kappa_2$ is equivalent to finding the ML estimator of $m_{ui}|n_{ui} \sim \operatorname{Poisson} (n_{ui}\kappa_2(-\log(1-e^\theta))$. This leads to the following update:
\bal{
	& \kappa_2 \leftarrow \dfrac{1}{-\log (1-e^\theta)} \dfrac{\sum_{ui}\Eq{m_{ui}}}{\sum_{ui}\Eq{n_{ui}}},
}
where $\Eq{m_{ui}} = \sum_{n=1}^{y_{ui}} n\kappa_2 \big( \Psi(y_{ui} - n + n\kappa_2) - \Psi(n\kappa_2)\big)~q(n_{ui}=n)$.

\section{EXPERIMENTAL RESULTS} \label{sec:expe}

\paragraph{Datasets.}
We consider the following datasets, whose structure is summarized in Table~\ref{tab:datasets}.

$\bullet$ The Taste Profile dataset \citep{bertin-mahieux_million_2011} provided by The Echo Nest contains the listening history of users. The data collected are the number of time the users listened to songs. We pre-process a subset of the data as in \citep{liang_modeling_2016}, keeping only users and songs that have more than 20 interactions. The histogram of the listening counts is displayed on Figure~\ref{fig:PPC}.

$\bullet$ The Last.fm dataset \citep{celma_music_2010} contains the listening history of users with additional timestamps information. We select play counts of the year 2008 and apply the same pre-processing as with the Taste Profile dataset.

$\bullet$ The NIPS dataset \citep{perrone_poisson_2016} contains bag-of-words representations of conference papers published in the NIPS conference from 1987 to 2015. We make an analogy between ``users who listened to songs'' and ``documents written with words''. The goal here is to recommend unused words to the author of a paper. 

\begin{table}[t] \footnotesize
 \caption{Datasets structure after pre-processing.}
 \begin{center}
 \begin{tabular}{llll} 
  \toprule
   & \textbf{Taste Profile} & \textbf{NIPS}  & \textbf{Last.fm} \\
  \midrule
  $\#$ rows 		& $16,301$		& $5,811$ 		& $781$ 	\\
  $\#$ columns	 	& $12,118$ 		& $11,463$ 		& $11,172$ 	\\
  $\#$ non-zeros 	& $1,176,086$ 	& $4,033,830$ 	& $402,058$ \\
  $\%$ non-zeros 	& $0.60\%$ 		& $6.06\%$ 		& $4.61\%$ \\
  \bottomrule
 \end{tabular}
 \end{center}
 \label{tab:datasets}
\end{table}

\paragraph{Experimental setup.}
\begin{table*}[ht!] \scriptsize
\caption{Recommendation performance of dcPF and PF using the Taste Profile dataset. Italic: scores of dcPF when using a grid-search for $\theta$, see text for details. Bold: two best NDCG scores (grid-search excluded).}
\begin{center}
\begin{tabular}{llllllll}
		\toprule
		\bf Model & 
			\bf Est. &
			$p=e^\theta$ &
			$\kappa$ & 
			$\bf NDCG0$ & 
			$\bf NDCG1$ &
			$\bf NDCG2$ &
			$\bf NDCG5$ \\
		\hline
		\multirow{2}{*}{\bf Log} & 
			VBEM&
			$0.803$& 
			$1$ & 
			$\bf 0.200~(\pm 3.0~10^{-3} )$& 
			$\bf 0.182~(\pm 2.3~10^{-3} )$&
			$\bf 0.166~(\pm 2.0~10^{-3} )$&
			$0.147~(\pm 1.5~10^{-3} )$\\
		& 
			\it Grid&
			$\it 0.3$& 
			$1$ & 
			$\it 0.200~(\pm 4.1~10^{-3})$ & 
			$\it 0.186~(\pm 3.9~10^{-3})$ &
			$\it 0.173~(\pm 3.7~10^{-3})$ &
			$\it 0.158~(\pm 3.7~10^{-3})$ \\
		\multirow{2}{*}{\bf ZTP} & 
			VBEM&
			$1.950$& 
			$1$ & 
			$0.192~(\pm 4.1~10^{-3})$& 
			$0.178~(\pm 3.7~10^{-3})$&
			$\bf 0.167~(\pm 3.6~10^{-3})$&
			$\bf 0.156~(\pm 3.8~10^{-3})$\\
		& 
			\it Grid&
			$\it 1$& 
			$1$ & 
			$\it 0.190~(\pm 3.5~10^{-3})$ & 
			$\it 0.178~(\pm 3.0~10^{-3})$ &
			$\it 0.168~(\pm 3.1~10^{-3})$ &
			$\it 0.158~(\pm 3.3~10^{-3})$ \\
		\multirow{2}{*}{\bf Geo} & 
			VBEM&
			$0.600$& 
			$1$ & 
			$0.199~(\pm 2.3~10^{-3})$& 
			$\bf 0.182~(\pm 1.8~10^{-3})$&
			$\bf 0.167~(\pm 1.8~10^{-3})$&
			$\bf 0.150~(\pm 1.2~10^{-3})$\\
		& 
			\it Grid&
			$\it 0.3$ & 
			$1$ & 
			$\it 0.199~(\pm 4.9~10^{-3})$ & 
			$\it 0.185~(\pm 4.6~10^{-3})$ &
			$\it 0.172~(\pm 4.2~10^{-3})$ &
			$\it 0.159~(\pm 4.0~10^{-3})$ \\
		\multirow{1}{*}{\bf Sh. NB} & 
			VBEM&
			$0.873$& 
			$(1,0.2)^T$& 
			$\bf 0.201~(\pm 3.1~10^{-3})$& 
			$\bf 0.183~(\pm 2.5~10^{-3})$&
			$\bf 0.166~(\pm 2.2~10^{-3})$&
			$0.147~(\pm 1.5~10^{-3})$\\
		\hline
		\multirow{1}{*}{\bf PFraw} & 
			. &
			. & 
			. & 
			$0.156~(\pm 3.0~10^{-3})$& 
			$0.155~(\pm 3.3~10^{-3})$&
			$0.150~(\pm 3.5~10^{-3})$&
			$0.144~(\pm 5.3~10^{-3})$\\
		\multirow{1}{*}{\bf PFbin} & 
			. &
			. & 
			. & 
			$0.197~(\pm 2.1~10^{-3})$& 
			$0.177~(\pm 1.5~10^{-3})$&
			$0.160~(\pm 1.5~10^{-3})$&
			$0.140~(\pm 1.3~10^{-3})$\\
		\bottomrule
	\end{tabular}
	\end{center}
	\label{tab:score_tps}
\end{table*}

\begin{table} \scriptsize
\caption{Performance with NIPS and Last.fm datasets.}
\begin{center}
\begin{tabular}{lllll}
	\toprule
	& 
		\multicolumn{2}{c}{\bf NIPS} &
		\multicolumn{2}{c}{\bf Last.fm} \\
	{\bf Model} & 
		$\bf NDCG0$ & 
		$\bf NDCG1$ &
		$\bf NDCG0$ &
		$\bf NDCG1$ \\
	\hline
	\multirow{1}{*}{\bf Log} &
		$\bf 0.394$ & 
		$\bf 0.430$ &
		$\bf 0.142$ &
		$\bf 0.129$ \\
	\multirow{1}{*}{\bf ZTP} & 
		$0.381$ & 
		$0.422$ &
		$0.122$ &
		$0.113$ \\
	\multirow{1}{*}{\bf Geo} & 
		$0.390$ & 
		$0.429$ &
		$0.139$ &
		$0.128$ \\
	\multirow{1}{*}{\bf Sh. NB} & 
		$\bf 0.396$ & 
		$\bf 0.431$ &
		$\bf 0.143$ &
		$\bf 0.130$ \\
	\hline
	\multirow{1}{*}{\bf PFraw} & 
		$0.358$ & 
		$0.405$ &
		$0.091$ &
		$0.088$ \\
	\multirow{1}{*}{\bf PFbin} & 
		$0.378$ & 
		$0.392$ &
		$0.122$ &
		$0.108$ \\
	\bottomrule
\end{tabular}
\end{center}
	\label{tab:score_other}
\end{table}

Each dataset is split into a train set $\Y^\text{train}$ containing $80\%$ of the non-zero of the original dataset and a test set $\Y^\text{test}$ containing the remaining $20\%$ (these values being set to $0$ in the train set). All the compared algorithms are trained with the train set and provide a recommendation list for each user. These lists are evaluated with the test set.

For each user, we recommend an ordered list of $m$ songs he/she never listened to, based on $\Y^\text{train}$. The songs in this list are sorted w.r.t. the prediction score defined by: 
${s_{ui} = \sum_{k=1}^K\Eq{w_{uk}}\Eq{h_{ik}}}$.
Note that, in dcPF, the expected number of listening sessions is equal to $[\W\H^T]_{ui}$, expressing long-term preference (see Section~\ref{sec:model}).

The quality of the proposed list is measured by the normalized discounted cumulative gain (NDCG) score. NDCG is a metric often used in information retrieval to evaluate lists of ranked items. It is calculated as follows:
\bal{
	\operatorname{DCG}_u = \sum_{i=1}^m \dfrac{{\operatorname{rel}(u,i)}}{\log_2(i+1)},~\operatorname{NDCG}_u = \frac{\operatorname{DCG}_u}{\operatorname{IDCG}_u}, \notag
}  
where $\operatorname{DCG}_u$ is the discounted cumulative gain and $\operatorname{rel}(u,i)$ is the relevance to the ground-truth of the $i$th item of the proposed list. $\operatorname{IDCG}_u$ is the ideal DCG, i.e., the best DCG score that can be obtained. Therefore $\operatorname{NDCG}_u \in [0,1]$ with $\operatorname{NDCG}_u=1$ corresponding to the perfect recommended list.
For the relevance to the ground-truth, we propose to account for the values in the test set above a fixed threshold: $\operatorname{rel}(u,i)= \mathds{1}[y^{\text{test}}_{ui}> s]$.
As mentioned in \citep{hu_collaborative_2008}, small listening counts reflect a preference with little confidence. For example, a user can listen to a song by pure curiosity. Therefore, this threshold leads to a more robust NDCG metric. We denote by $\operatorname{NDCG\text{$s$}}$ the NDCG with the threshold $s$. If $s=0$, we recover the classic $\operatorname{NDCG0}$ metric for binary data. Otherwise, for $s\geq1$, $\operatorname{NDCG\text{$s$}}$ only considers the songs which have been listened to at least $s$ times and ignores the others. Other metrics such as precision and recall lead to similar conclusions than $\operatorname{NDCG0}$ and will not be displayed in the following.

\paragraph{Compared methods.}
For the three datasets, we compare dcPF with its limit cases: PF performed on raw (PFraw) or on binarized data (PFbin). PF is known to achieve good performance in recommendation tasks \citep{gopalan_scalable_2013}. For the Taste Profile dataset, we considered the hyper-parameters $\alpha^W=\alpha^H$ among $\{0.1,0.3,1\}$ and $K$ among $\{50,100,200\}$, and selected the values $\alpha^W=\alpha^H=0.3$ and $K=100$ which gave the best $\operatorname{NDCG0}$ for PFbin. For the NIPS and Last.fm datasets, which are smaller than the Taste Profile dataset, we only considered $\alpha^W=\alpha^H=0.3$ and $K=50$. The stopping criterion of the algorithms is set to $\tau = 10^{-5}$. Evaluation is done using a ranked list of $m=100$ items. For all the experiments, algorithms are run five times from random initializations.\footnote{Algorithms and experiments are available on github: \url{https://github.com/Oligou/dcPF}.}

\paragraph{Prediction results.} \label{sec:prediction_result}
We start by discussing results with the Taste Profile dataset, reported in Table~\ref{tab:score_tps}. A general observation is that dcPF gives better results than the two baselines PFraw and PFbin for all four metrics and every element distribution, with the exception of ZTP in the case of NDCG0. PFbin returns better scores than PFraw up to $s=5$. This confirms the usefulness of the binarization stage when using PF, but only up to a certain threshold $s$ (this is because PFbin does not fully exploit the non-zero values in the original raw data). The performance gap between dcPF and PFbin increases with the threshold $s$. On the two other datasets (NIPS and Last.fm), Table~\ref{tab:score_other} shows that dcPF outperforms the baseline methods for all element distributions. Note that for the NIPS dataset, a somehow different context from song recommendation, PFraw is effective and performs better than PFbin as soon as the threshold $s$ is larger than one. From both Tables~\ref{tab:score_tps} and \ref{tab:score_other} we conclude that the proposed shifted NB element distribution is a good compromise overall.

\paragraph{Natural parameter estimation.} As explained in Section~\ref{sec:tradeoff}, estimation of the natural parameter tells us about the level of scale information exploited by dcPF. Table~\ref{tab:score_tps} shows that dcPF indeed offers a valuable trade-off between PFraw and PFbin, because the estimated parameter $\theta$ lies in between the two limit cases $\theta^\text{raw}$ and $\theta^\text{bin}$. To assess the quality of the estimation procedures for $\theta$ in Log, ZTP and Geo described in Section~\ref{sec:tradeoff} (plainly referred as VBEM), Table~\ref{tab:score_tps} also displays evaluation metrics obtained with a grid-search. More precisely, we use $\theta=\log p$ which maximizes NDCG5 from a set of pre-specified values. For Log and Geo, $p$ is searched between $0$ and $1$ with a step of $0.1$. For ZTP, $p$ is searched in $\{0,0.1,0.5,1,2,10,100,+\infty\}$. It appears that VBEM slightly over-estimates the optimal value (in terms of NDCG5) of the natural parameter, but remains a very robust procedure.

\paragraph{Posterior predictive check (PPC).}
\begin{figure}
	\begin{center}
	\includegraphics[height=3.75cm]{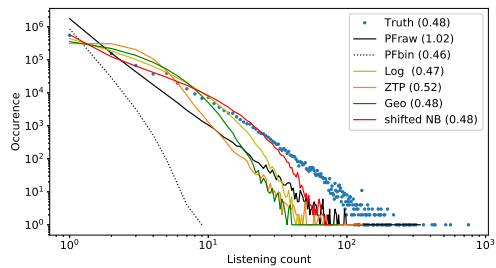}
	\end{center}
	\caption{PPC of the distribution of the non-zero values in the Taste Profile dataset. The blue points (Truth) represents the histogram of the non-zero values in the train set. The colored curves represent the simulated histograms obtained from the different inferred dcPF or PF models. The percentages of non-zero values are written in parentheses.}
	\label{fig:PPC}
\end{figure}

We provide a PPC of the distribution of the listening counts in the Taste Profile dataset (see Figure~\ref{fig:PPC}). A PPC consists in simulating a new dataset $\Y^\text{PPC}$ from the fitted model (for dcPF, we simulate from the generative process described in Section~\ref{sec:model} with latent variables $\W$, $\H$ and parameters inferred in Section~\ref{sec:prediction_result}). Then, we compare the histogram of the values of $\Y^\text{train}$ and $\Y^\text{PPC}$. 
The PPC of the two limit cases is very instructive. PFraw tries to fit the long tail of the data, but, by doing so, destroys the representation of the zero values ($1.02\%$ of non-zero values versus $0.48\%$ in the real dataset). It can explain the disappointing performances of PFraw for NDCG0. On the contrary, PFbin better fits the sparsity of the data but is not able to describe large counts. In both cases, PF struggles to properly weigh the influence of large counts compared to low counts. Comparatively, dcPF proposes a smoother weighting between large and low values. dcPF respects  both the sparsity and the long tail of the data for the four element distributions. ZTP seems to over-estimate the influence of medium counts (from 1 to 5), whereas shifted NB has the best fit to the histogram. We observe that regardless of the model, explaining the large values ($>100$) remains difficult, however we may consider that after a certain threshold the counts do not contain useful information.

\section{CONCLUSION} \label{sec:conc}

In this paper, we described new contributions to cPF for discrete data. As compared to PF, we showed that dcPF offers more flexibility to model long-tailed data. Inference remains scalable thanks to modeling of the non-zero values only. Numerical experiments confirmed that our adaptive VBEM algorithm efficiently exploits raw data, leading to better recommendation scores when compared to the two limit cases (PF on raw and binarized data). Among the four element distributions presented and experimented in this work, the proposed shifted NB prove particularly efficient thanks to its additional parameter, and often led to the best recommendation scores.

Based on this work, a number of exciting perspectives can be considered, such as investigating more complex element distributions to better fit the extreme observations or to address other forms of data. For instance, it would be of great interest to adapt the model for bounded data such as ratings, which are widespread in CF but cannot be processed with dcPF in its current form.

\paragraph{Acknowledgements.}
Supported by the European Research Council (ERC FACTORY-CoG-6681839) and the ANR-3IA (ANITI).


\appendix

\section{Appendix: Stirling Numbers}

The Stirling numbers of the three kinds are three different ways to partition $y$ elements into $n$ groups.

$\bullet$ The Stirling number of the first kind corresponds to the number of ways of partitioning $y$ elements into $n$ disjoints cycles.

$\bullet$ The Stirling number of the second kind corresponds to the number of ways of partitioning $y$ elements into $n$ non-empty subsets. 

$\bullet$ The Stirling number of the third kind (also known as Lah number) corresponds to the number of ways of partitioning $y$ elements into $n$ non-empty ordered subsets. 

\begin{figure}[h!]
\tikzset{
  font={\fontsize{4pt}{6}\selectfont}}

\begin{tabular}{c|c|c}

{\bf First kind} & {\bf Second kind} & {\bf Third kind} \\
{$St_1(3,1)=2$} & {$St_2(3,1)=1$} & {$St_3(3,1)=6$} \\
&& \\

\begin{tikzpicture}
	\tikzstyle{main}=[circle, inner sep=0pt, minimum size=7pt, thick, draw =black!80, node distance = 2.5mm]
	\tikzstyle{connect}=[-latex, thick]
	\tikzstyle{box}=[rectangle, draw=black!100]

	\node[main, fill = red] (a) [] { };
	\node[main, fill = green] (b) [right=of a] { };
	\node[main, fill = blue] (c) [below=of a] { };

	\path (a) edge [connect] (b)
	(b) edge [connect] (c)
	(c) edge [connect] (a);
\end{tikzpicture}

&

\begin{tikzpicture}
	\tikzstyle{main}=[circle, inner sep=0pt, minimum size=7pt, thick, draw =black!80, node distance = 2.5mm]
	\tikzstyle{connect}=[-latex, thick]
	\tikzstyle{box}=[rectangle, draw=black!100]

	\node[main, fill = red] (a) [] { };
	\node[main, fill = green] (b) [right=of a] { };
	\node[main, fill = blue] (c) [below=of a] { };

	\path (a) edge (b)
	(b) edge (c)
	(c) edge (a);
\end{tikzpicture}

&
\begin{tikzpicture}
	\tikzstyle{main}=[circle, inner sep=0pt, minimum size=7pt, thick, draw =black!80, node distance = 2.5mm]
	\tikzstyle{connect}=[-latex, thick]
	\tikzstyle{box}=[rectangle, draw=black!100]

	\node[main, fill = red] (a) [] { };
	\node[main, fill = green] (b) [right=of a] { };
	\node[main, fill = blue] (c) [below=of a] { };

	\path (a) edge [connect] (b)
	(b) edge [connect] (c);
\end{tikzpicture}
\hspace{2pt}
\begin{tikzpicture}
	\tikzstyle{main}=[circle, inner sep=0pt, minimum size=7pt, thick, draw =black!80, node distance = 2.5mm]
	\tikzstyle{connect}=[-latex, thick]
	\tikzstyle{box}=[rectangle, draw=black!100]

	\node[main, fill = red] (a) [] { };
	\node[main, fill = green] (b) [right=of a] { };
	\node[main, fill = blue] (c) [below=of a] { };

	\path (b) edge [connect] (a)
	(c) edge [connect] (b);
\end{tikzpicture}
\\
\begin{tikzpicture}
	\tikzstyle{main}=[circle, inner sep=0pt, minimum size=7pt, thick, draw =black!80, node distance = 2.5mm]
	\tikzstyle{connect}=[-latex, thick]
	\tikzstyle{box}=[rectangle, draw=black!100]

	\node[main, fill = red] (a) [] { };
	\node[main, fill = green] (b) [right=of a] { };
	\node[main, fill = blue] (c) [below=of a] { };

	\path (b) edge [connect] (a)
	(c) edge [connect] (b)
	(a) edge [connect] (c);
\end{tikzpicture}
& &
\begin{tikzpicture}
	\tikzstyle{main}=[circle, inner sep=0pt, minimum size=7pt, thick, draw =black!80, node distance = 2.5mm]
	\tikzstyle{connect}=[-latex, thick]
	\tikzstyle{box}=[rectangle, draw=black!100]

	\node[main, fill = red] (a) [] { };
	\node[main, fill = green] (b) [right=of a] { };
	\node[main, fill = blue] (c) [below=of a] { };

	\path (b) edge [connect] (c)
	(c) edge [connect] (a);
\end{tikzpicture}
\hspace{2pt}
\begin{tikzpicture}
	\tikzstyle{main}=[circle, inner sep=0pt, minimum size=7pt, thick, draw =black!80, node distance = 2.5mm]
	\tikzstyle{connect}=[-latex, thick]
	\tikzstyle{box}=[rectangle, draw=black!100]

	\node[main, fill = red] (a) [] { };
	\node[main, fill = green] (b) [right=of a] { };
	\node[main, fill = blue] (c) [below=of a] { };

	\path (c) edge [connect] (b)
	(a) edge [connect] (c);
\end{tikzpicture}
\\
& &
\begin{tikzpicture}
	\tikzstyle{main}=[circle, inner sep=0pt, minimum size=7pt, thick, draw =black!80, node distance = 2.5mm]
	\tikzstyle{connect}=[-latex, thick]
	\tikzstyle{box}=[rectangle, draw=black!100]

	\node[main, fill = red] (a) [] { };
	\node[main, fill = green] (b) [right=of a] { };
	\node[main, fill = blue] (c) [below=of a] { };

	\path (a) edge [connect] (b)
	(c) edge [connect] (a);
\end{tikzpicture}
\hspace{2pt}
\begin{tikzpicture}
	\tikzstyle{main}=[circle, inner sep=0pt, minimum size=7pt, thick, draw =black!80, node distance = 2.5mm]
	\tikzstyle{connect}=[-latex, thick]
	\tikzstyle{box}=[rectangle, draw=black!100]

	\node[main, fill = red] (a) [] { };
	\node[main, fill = green] (b) [right=of a] { };
	\node[main, fill = blue] (c) [below=of a] { };

	\path (b) edge [connect] (a)
	(a) edge [connect] (c);
\end{tikzpicture}

\end{tabular}
\caption{Illustration of the Stirling numbers of the three kinds for $y=3$ and $n=1$.}
\label{fig:stirling}
\end{figure}
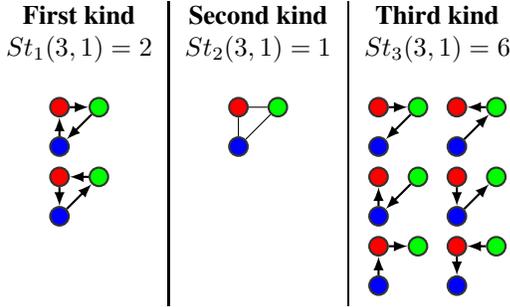

\section{Appendix: Proof of limit cases}

\begin{prop}
	If there exists $\theta^\text{raw}$ such that $\lim_{\theta\to\theta^\text{raw}}\kappa^T \psi(\theta) = -\infty$, then the posterior of dcPF tends to the posterior of PF as $\theta$ goes to $\theta^\text{raw}$.
\end{prop}

\begin{prop}
	If there exists $ \theta^\text{bin}$ such that $\lim_{\theta\to\theta^\text{bin}}\kappa^T \psi(\theta) = +\infty$, then the posterior of dcPF tends to the posterior of PF applied to binarized data as $\theta$ goes to $\theta^\text{bin}$, i.e.: 
		${\lim_{\theta\to\theta^\text{bin}} p(\W,\H|\Y) = p(\W,\H|\N=\Y^b)}$.
\end{prop}

\begin{proof}
	Let $\lambda\in\mathbb{R}_+$, $n\sim\operatorname{Poisson}(\lambda)$
	and $y|n\sim ED(\theta,n\kappa)$ with support given by $S = \{n,\dotsc,+\infty\}$: 
	\bal{
		&p(n|\lambda) = \dfrac{\lambda^n e^{-\lambda}}{n!}, \\
		&p(y|n) = \exp(y\theta- n\kappa^T\psi(\theta))h(y,n\kappa),~y\in S,
	} 
	where $\kappa$ and $\psi(\theta)$ can either be scalars or vectors of the same dimension. In both cases, $\kappa^T\psi(\theta) \in \mathbb{R}$. We denote by $r=\lambda e^{-\kappa^T \psi(\theta)}$. 

	We have the following posterior distribution for $y>0$:
	\bal{
		&p(n|y) = \dfrac{r^n h(y,n\kappa)(n!)^{-1}}{\sum_{m=1}^{y} r^m h(y,m\kappa)(m!)^{-1}},~n\in\{1,\dotsc,y\}.
	}
	Thus, for fixed $\kappa$ and $y>0$, we have that:
	\bal{
		\sum_{m=1}^{y} r^m h(y,m\kappa)(m!)^{-1} 
		&\underset{r\to+\infty}{\sim} r^y h(y,y\kappa)(y!)^{-1}\\
		&\underset{r\to 0 }{\sim} r h(y,\kappa).
	}
	It follows:
	\bal{
		&p(n|y) \xrightarrow[r\to+\infty]{} \delta_y(n) \\
		&p(n|y) \xrightarrow[r\to 0]{} \delta_1(n).
	}

	From these results we can deduce that, in dcPF, assuming:

	$\bullet$  there exists $\theta^\text{raw}$ such that $\lim_{\theta\to\theta^\text{raw}}\kappa^T \psi(\theta) = -\infty$,

	$\bullet$  there exists $\theta^\text{bin}$ such that $\lim_{\theta\to\theta^\text{bin}}\kappa^T \psi(\theta) = + \infty$.

	Then, we have the following limit cases:
	\bal{
		&p(\N|\Y) =\int\limits_{\W,\H} p(\N|\Y,\W,\H)p(\W,\H|\Y) d\W d\H \notag\\
		& \xrightarrow[\theta\to\theta^\text{raw}]{} \int\limits_{\W,\H} \delta_{\Y}(\N)~p(\W,\H|\Y) d\W d\H = \delta_{\Y}(\N) \notag\\
		& \xrightarrow[\theta\to\theta^\text{bin}]{} \int\limits_{\W,\H} \delta_{\Y^b}(\N)~p(\W,\H|\Y) d\W d\H =  \delta_{\Y^b}(\N).
	}
	And finally, for the posterior distribution: 
	\bal{
		p(\W,\H|\Y)&=\int\limits_{\N} p(\W,\H|\N)p(\N|\Y)d\N \\
		& \xrightarrow[\theta\to\theta^\text{raw}]{} p(\W,\H|\N=\Y) \\
		& \xrightarrow[\theta\to\theta^\text{bin}]{} p(\W,\H|\N=\Y^b),
	}
	where $p(\W,\H|\N)$ is the posterior of a PF model with raw or binarized observations respectively.
\end{proof}

\section{Appendix: Adaptivity of dcPF to over-dispersion}

\begin{table}[h] \footnotesize
 \caption{Mean, variance and ratio var/mean of the non-zero values for each dataset.
 Learned parameters for each model and each dataset.}
 \begin{center}
 \begin{tabular}{llll} 
  \toprule
   & \textbf{Taste Profile} & \textbf{NIPS}  & \textbf{Last.fm} \\
  \midrule
  mean of non-zeros & $2.66$ & $2.74$ & $3.86$ \\
  var of non-zeros & $25.94$ & $20.87$ & $65.72$ \\
  ratio var/mean & $9.8$ & $7.6$ & $17.0$ \\
  \midrule
  Log - $p$ & $0.80$ & $0.74$ & $0.90$ \\
  ZTP - $p$ & $1.95$ & $1.40$ & $2.35$ \\
  Geo - $p$ & $0.60$ & $0.51$ & $0.69$ \\
  sh. NB - $p$ & $0.87$ & $0.86$ & $0.90$ \\
  sh. NB - $\kappa_2$ & $0.21$ & $0.17$ & $0.27$ \\
  \bottomrule
 \end{tabular}
 \end{center}
 \label{tab:dispersion property}
\end{table}

Table~\ref{tab:dispersion property} illustrates how the natural parameter $\theta=\log(p)$ is strongly correlated to the variance-mean ratio of the non-zero values of the datasets. Hence, it illustrates the adaptivity of dcPF to over-dispersion.

\subsubsection*{References}

\bibliographystyle{ACM-Reference-Format}
\bibliography{2019_UAI}

\end{document}